\documentclass[%
 reprint,
 superscriptaddress, 
 amsmath,amssymb,
 aps,
 pra,
 floatfix,
]{revtex4-2}

\usepackage{graphicx}
\usepackage{dcolumn}
\usepackage{bm}
\usepackage{dsfont}
\usepackage[dvipsnames]{xcolor}
\usepackage{enumitem}
\usepackage{multirow}
\usepackage[colorlinks=true,citecolor=Cerulean,linkcolor=RubineRed,urlcolor=Cerulean]{hyperref}
\usepackage{array}
\usepackage{float}
\usepackage{csquotes}
\usepackage[T1]{fontenc}
\usepackage[normalem]{ulem}
\MakeOuterQuote{"}

\usepackage[textsize=tiny]{todonotes}

\definecolor{customcolorblue}{HTML}{4573ae}

\usepackage{mathtools}

\usepackage{algpseudocode}
\usepackage{algorithm}

\usepackage{tikz}
\usepackage{amsthm}
\usepackage{amssymb}
\newtheorem{theorem}{Theorem}

\newtheorem{lemma}[theorem]{Lemma}

\theoremstyle{definition}

\usepackage{braket}

\begin{document}

\title{Quantum Hamlets: Distributed Compilation of Large Algorithmic Graph States}

\author{Anthony Micciche}
\email{amicciche@umass.edu}
\affiliation{Manning College of Information and Computer Sciences, University of Massachusetts Amherst, 140 Governors Drive
Amherst, Massachusetts 01003, USA}

\author{Naphan Benchasattabuse}%
\affiliation{Graduate School of Media and Governance, Keio University, Fujisawa, Kanagawa 252-0882, Japan}%

\author{Andrew McGregor}
\affiliation{Manning College of Information and Computer Sciences, University of Massachusetts Amherst, 140 Governors Drive
Amherst, Massachusetts 01003, USA}

\author{Michal Hajdu\v{s}ek}%
\affiliation{Graduate School of Media and Governance, Keio University, Fujisawa, Kanagawa 252-0882, Japan}%

\author{Rodney Van Meter}%
\affiliation{Graduate School of Media and Governance, Keio University, Fujisawa, Kanagawa 252-0882, Japan}
\affiliation{Faculty of Environment and Information Studies, Keio University, 5322 Endo, Fujisawa, Kanagawa 252-0882, Japan}%

\author{Stefan Krastanov}
\affiliation{Manning College of Information and Computer Sciences, University of Massachusetts Amherst, 140 Governors Drive
Amherst, Massachusetts 01003, USA}

\date{\today}
 
\begin{abstract}
We investigate the problem of compiling the generation of graph states to arbitrarily many distributed homogeneous quantum processing units (QPUs), providing a scalable partitioning algorithm and graph state generation protocol to minimize the number of Bell pairs required. Current approaches focus on the naive metric of cut edges to estimate the quantum communication cost. We show that the problem of balanced $k$ graph partitioning, with the objective of minimizing the sizes of the maximum matchings between the partitions, leads to lower entanglement requirements across partitions. Our heuristic algorithm, BURY, partitions graph states to require fewer Bell pairs for generation than state-of-the-art $k$ partition algorithms. Furthermore, we show that BURY reduces the cut-rank of the partitions, demonstrating that the partitioning found by our algorithm is likely to minimize the Bell pair utilization of any future improved distributed graph state generation protocol. We also discuss how our methods apply to the dynamic case where the graph state generation and measurement are performed concurrently. Our compilation approach provides a scalable foundation for reducing quantum network overhead for distributed measurement-based quantum computation (MBQC), as well as any scheme where distributed graph state generation is desired.
\end{abstract}

\maketitle

\section{Introduction}
In recent years, it has become increasingly apparent that useful quantum computers will need to be distributed in some fashion in order to scale to practically large sizes \cite{path_to_qc}. With this comes the problem of distributed circuit compilation. That is, the problem of mapping the gates and addressed logical qubits of a hardware-agnostic quantum program or circuit to a constrained physical system. In particular, when compiling to a quantum-networked distributed system, the primary bottleneck will be the number of times that long-range entanglement is required between the QPUs \cite{optimized_dqc_comiler}. So much so that QPU-local two-qubit gates can be viewed as "free". Our work follows in the trend of wanting to reduce the number of long-range entangled states required. However, while the majority of existing work on distributed compilation has focused on the circuit model of quantum computation, instead we explore the measurement-based quantum computation (MBQC) model \cite{one_way,mbqc}. The main benefit of the MBQC model is that instead of performing unitary evolution as in the circuit-based model of computation (CBQC), one \textit{only} needs to generate a
graph state and then perform single qubit measurements to achieve universality \cite{one_way}. Because both are equivalent, whether MBQC or CBQC is advantageous will depend on hardware specific error rates and the differences between gate and measurement speeds for the specific hardware in question.
One potential downside of MBQC is that it can require more qubits. However, this can be addressed by preparing the graph state as it is being measured, as well as through optimized compilers to reduce the sizes of the required graph state to some compiled algorithmic graph states as in \cite{Jabalizer}. Despite the capability for dynamic generation of these graph states in MBQC, the spatial cost still grows with the path-width of the desired graph state \cite{Elman_2025}. Because pathwidth can grow even linearly with the input size (for a very simple example, $K_n$, a complete $n$ vertex graph), to achieve certain large scales of MBQC (especially on a "matter qubit" platform where qubit-number-per-module can be a limiting factor due to fabrication constraints) distributed quantum computing is inevitable. 

 In this work, we propose a scalable algorithm for partitioning an arbitrary graph state into an arbitrary number of pieces, each piece corresponding to a QPU. The goal of the algorithm is to minimize the number of long-range Bell pairs required to realize a single large graph state spanning a multitude of QPUs. To make progress on the larger problem of designing schemes for distributed MBQC, we address the problem of graph state partitioning in the case where the entire graph state is generated before measurement, although we discuss in Section~\ref{sec:dynamic_compilation} how our work could be expanded to the dynamic generation case. Because of this assumption, our work could also be applied to any regime where one would want to generate a distributed graph state.

In Section~\ref{sec:background}, we provide a brief review of the motivating literature and related works, as well as a mathematical background section to review how graph states can be manipulated and thought of from the perspective of standard graph algorithm design. We intend this section to be for the purpose of making this work more self-contained and accessible to scientists with only modest knowledge of quantum computing, assuming they are comfortable with graph theory. Following this, Section~\ref{sec:quantum_land} sets the stage upon which the rest of this work will be carried out, describing our assumptions, abstractions, and nomenclature.

Section~\ref{sec:protocol} describes protocols for generating an arbitrary graph state using only local CZ gates, $Y$ measurements, and long-range logical Bell pair generations. We discuss here why minimizing the number of edges cut fails (such as with an algorithm like Kernighan–Lin \cite{kl_algo}) for graph state partitioning, and why our generation protocol, \textit{vertex cover grafting} (VCG) better embodies the goal of minimizing entanglement between partitions. Our VCG protocol requires a number of Bell pairs equal to the sizes of the matchings across all partitions, and we find that algorithm design for this metric also reduces other measures of entanglement, namely cut rank (discussed in Appendix~\ref{sec:cutrank}). 

Our main contribution is discussed in Section~\ref{sec:algorithm}, where we introduce and describe our heuristic algorithm, BURY, that we develop to reduce the number of required long-range logical Bell pairs consumed by VCG, and moreover reduce entanglement entropy between partitions. Following this, Section~\ref{sec:results} compares our method to heuristically applied existing algorithms that were designed for minimizing the total number of edges cut by the $k$-partition. We show that our method outperforms nearly all existing methods on most graphs tested. We find only one other method (based on graph coarsening, METIS \cite{KARYPIS199896}) which even approaches the performance of our heuristic algorithm, beating us only on a few synthetic micro-benchmarks unrelated to real algorithmic graph states. While METIS has already been used for inspiration of circuit partitioning methods \cite{Burt2026multilevelframework}, seeing its performance as a graph coarsening based method performing well in our context is notable for future algorithm development. 

Section~\ref{sec:dynamic_compilation} provides preliminary discussion for applying our work to the situation of consuming the graph state while generating it (dynamic compilation), as well as potential future work on the topic. Finally, Section~\ref{sec:conclusion} continues the discussion of potential future work beyond dynamic compilation.
In the Appendix, plots showing that our method works just as well for reducing cut rank are presented.
Additionally, we provide a proof of the NP hardness of finding a balanced partition of a graph such that the maximum matching between partitions is minimized, a problem that our heuristic algorithm aims to solve.

\section{Background\label{sec:background}}
\subsection{Previous work}

Previous work has investigated compilation of algorithmic graph states \cite{Jabalizer}, and the problem of scheduling generation of these graph states on non-distributed hardware \cite{Sitong_2023} using a Game of Surface Codes \cite{Litinski2019gameofsurfacecodes} style abstraction.

As the spatial cost of generating graph states scales with its path width~\cite{Elman_2025}, even in the most optimized MBQC model of computation, in order to realize large algorithmic graph states, scheduling the generation of such states in a distributed setting is an essential and underexplored area of research.
An area that this work addresses directly, where we refer to as distributed compilation the problem of assigning vertices in a given graph state to different QPUs in order to minimize the number of Bell pairs the network needs to generate.

Other work has addressed various distributed graph state generation protocols \cite{Meignant2019, fischer_towsley, ji2025distributingarbitraryquantumgraph}, but does not consider the problem of vertex to QPU assignment for arbitrarily large graphs. Another issue with existing work on graph state generation is its reliance on non-scalable algorithms, such as using integer-linear programs \cite{Sharma_2026}, which have exponential running time with the problem size. The existing literature on partitioning graph states (equivalent to vertex to QPU assignment) is quite sparse, or tends to rely on application of algorithms that were designed with minimizing edge counts between partitions \cite{KokoriThesis}. We show that minimizing the number of edges is far from an ideal proxy for minimizing the number of required Bell pairs. Our algorithm minimizes maximum matching sizes and performs better in terms of both Bell pair consumption necessitated by our proposed graph state generation protocol, as well as reducing cut rank, a more universal measure of entanglement. 

During the preparation of this manuscript, a concurrent work addressed the bipartition problem via a heuristic simulated annealed algorithm \cite{pettersen2026bipartitioninggraphstatesdistributed} to minimize cut rank between the two partitions. However, the authors did not state how their work could be applied to the case of more than two QPUs, and relied on an existing general purpose heuristic algorithm to minimize Bell pair consumption. In the case of two QPUs, the number of Bell pairs needed to generate the graph state is equal to the cut rank~\cite{fattal2004entanglementstabilizerformalism}.
However, this is not true for an arbitrary number of partitions, where the cut rank represents a lower bound on the required number of Bell pairs~\cite{hein2006}, and one needs to specify the protocol that will be used to generate the graph state.

We address the shortcomings of previous works by designing a scalable $O(|V|^2+|V|\Delta^2)$ time (Section~\ref{sec:time_complexity}) heuristic framework tailored specifically to graph state partitioning, which can easily be applied to partitioning large graph states into any number of pieces. Furthermore, while we provide and analyze a greedy algorithm invocation of this framework, we expect it to later be adapted to other more nuanced paradigms as well, such as constrained networks.

\subsection{Mathematical background\label{sec:background_math}}
Throughout this work, we consider a graph $G = (V,E)$ with $n$ vertices, and $m$ edges. This graph corresponds to a quantum state, namely a graph state \cite{hein2006}, which can be thought of 
as each vertex representing a qubit in the $\ket{+}$ state, and each edge corresponds to performing a CZ between them. Formally, for a graph $G$, we can interpret it as a graph state $\ket{G}$ by
$$\ket{G} := \prod_{(i,j) \in E}CZ_{i,j}\ket{+}^{|V|}.$$

From a classical graph theory point of view, the main functional difference between a graph and a graph state is how they transform under the action of \textit{local complementation} (LC).
Local complementation takes a graph $G=(V,E)$ and produces a new graph $G'=(V,E')$.
Despite the difference in edge sets $E$ and $E'$, the corresponding graph states $\ket{G}$ and $\ket{G'}$ are related by single-qubit Clifford unitaries.
These unitary operations are far more simple to implement than two-qubit gates, especially when these two-qubits act on qubits in different QPUs.
We make the standard assumption that two-qubit gates within a single QPU incur negligible cost compared to inter-QPU two-qubit gates~\cite{optimized_dqc_comiler}.
This is because inter-QPU two-qubit gates rely on the consumption of Bell pairs which are currently the main bottleneck in distributed quantum computation~\cite{Kimble08,WehnerElkoussHanson_qnet2018,main2024distributed}.

Given a graph $G$, we denote LC of a graph vertex $a$ by $\tau_a(G)$ 
$\tau_a(G)$ means to take the induced subgraph formed from the neighborhood of $a$ (excluding $a$), remove all present edges, and add all edges between vertices that were not in $G$ before. In other words, it is the element-wise NOT operation on the adjacency matrix specified on only the exclusive neighborhood of $a$. 

Local complementation alone is a powerful tool; however, in this work, we will always pair it with use of additional qubits (vertices) and quantum measurement. Since a graph state corresponds to a physical system, we can intermediately use additional vertices to help realize it. Removal of a vertex $a$ in $G$ is called $Z$ measurement, $Z_a(G)$, and performing a local complementation on $a$ before the $Z$ measurement is $Y$ measurement, $Y_a(G)$. A simple example would be performing $Y$ measurement on any vertex in a path graph with $n$ vertices, $P_n$, will yield $P_{n-1}$. From this example, notice that the LC included in the $Y$ measurement keeps the graph connected, whereas doing a $Z$ measurement on any non-end vertex would yield a disconnected graph.

As this work considers partitioning graph states (see Figure~\ref{fig:hamlet_abstraction} for a brief abstraction), additional qubits using $Y$ measurements are useful for manipulating the structure of edges crossing the cut. An example of this is shown in Figure~\ref{fig:vcg}, our description of our graph state generation protocol (elaborated on in Section~\ref{sec:protocol}). This figure shows that using two additional qubits and $Y$ measuring them, we can generate a $K_{1,n-1}$ biclique with only a single edge crossing a partition \cite{fischer_towsley}. Moreover, this same idea can be applied to generate an arbitrary $K_{a,b}$ biclique, using additional LCs to clean up unwanted edges, but this will not be used in this work. 

\section{A Quantum Land\label{sec:quantum_land}}
In this work, we define the term "hamlet" to refer to a QPU, and the villagers who live in a given hamlet correspond to logical qubits. Each hamlet will also have at least one mayor, corresponding to the logical communication qubit in more common nomenclature. The mayors are connected to a quantum network and we assume that we can place mayors into Bell states. We assume intra-village operations to be free, and inter-village communication to be costly. In other words, villagers from one village must go through their mayor in order to "talk" to a villager in another hamlet, but may talk freely among themselves. A sketch of this abstraction is present in Figure~\ref{fig:hamlet_abstraction}. 

\begin{figure}
    \centering
    \includegraphics[width=\linewidth]{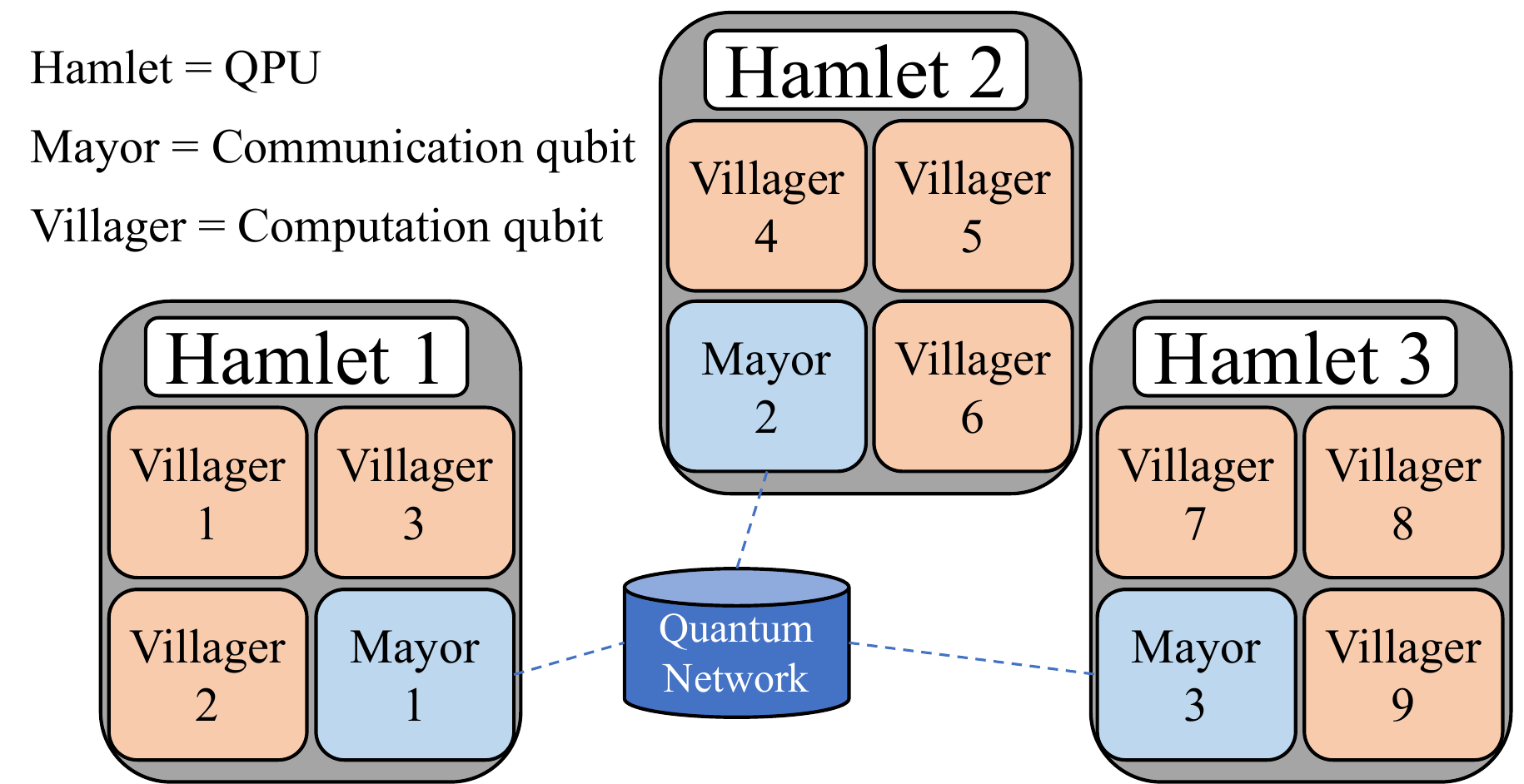}
    \caption{\textbf{Level of abstraction at which this work operates.} Each gray box in the figure corresponds to a different hamlet (QPU), in this example, each of which has three villagers (logical qubits), and one mayor (communication qubit). The mayor qubits are end nodes in some quantum network and therefore, for a cost, we can place two mayors into entangled state. We assume operations between two villagers within the same hamlet to be free, while an operation that involves an interaction between villagers in separate villages requires their mayors to be entangled.}
    \label{fig:hamlet_abstraction}
\end{figure}

When considering partitions of graph $G$, we use $k$ to denote the number of partitions being considered. We let $k$ correspondingly also refer to the number of hamlets, and $\frac{n}{k}$ is then the number of villagers per hamlet, as we assume each hamlet is the same size and, without loss of generality, assume that $k$ divides $n$ for conciseness. This divisibility assumption does not affect our algorithm we describe in Section~\ref{sec:algorithm}, as that only requires the sizes of the partitions as inputs. 

We address the problem of determining, for a certain number of these hamlets, and a provided arbitrary graph state, how to best \textit{assign} qubits (vertices) in the graph state to villagers in the hamlets model to minimize the number of times that in total we must entangle two mayors in order to generate the provided graph state. We call this problem \textbf{static compilation}, and with the exception of Section~\ref{sec:dynamic_compilation}, assume that all hamlets are homogeneous. Since all hamlets are equipped with the same resources, rather than thinking about qubit \textit{assignment}, we instead consider \textit{partitioning} the graph-state into $k$ pieces (when we have $k$ hamlets), under some metric that proxies the number of long range entang fements required. Graph state generation protocols and their associated metrics are discussed in the following section.

Our work formalizes the problem of generating distributed graph states with minimal entanglement resources in a way that enables us to use known graph partitioning algorithms. Moreover, we design a new partitioning algorithm in this manuscript that drastically lowers the inter-hamlet (i.e., cross-module) entanglement resources required to prepare a distributed graph state compared to all partition algorithms tested.

\section{Protocols for generating graph states\label{sec:protocol}}
In order to most efficiently partition a graph state, we will need to first define how the graph state will be generated for a provided partitioned graph state (e.g., define the "cost" of a partition). For example, if we will generate the graph state by only using $\ket{+}$ initializations and $CZ$ gates for all edges (using one Bell pair per inter-hamlet edge), then the problem is nothing more than balanced graph partitioning, with the traditional objective of minimizing cut edges.
This problem has been studied for decades \cite{kl_algo, saran_vazirani, andreev_racke_2004}, and is known to be NP-hard \cite{andreev_racke_2004}.
However, generating graph states edge by edge like this is extremely inefficient, so we instead provide a drastically more efficient protocol that is also quite simple. We call it \textit{vertex cover grafting} (VCG). Knowing how the graph state will be generated (i.e., with VCG) then allows us to address the design of algorithms for partitioning the graph state in order to minimize the cost of using the specified protocol. Moreover, although the heuristic we design in Section~\ref{sec:algorithm} aims to minimize Bell pair usage when generating graph states with VCG, we further show in Appendix~\ref{sec:cutrank} that our algorithm happens to also reduce universal measures of entanglement, namely cut rank. 

As alluded to in Section~\ref{sec:background_math}, an arbitrarily large biclique between two hamlets can be formed for the cost of a single long-range Bell pair. One may think that we could decompose the input graph into biclique components, and use those to generate the graph state (also assuming that the local complementation operations involved in such a protocol would not interfere, which is a non-trivial ask). This unfortunately is the bipartite dimension of the graph, and finding it has been proven to be NP-complete \cite{GareyJohnson1979}.

For guidance in the design of our heuristic, we instead look to the minimization of the cut rank, which is proportional to the rank of the adjacency matrix for the nodes in the graph state with edges crossing the partition \cite{hein2006}. During the finalization of this manuscript, this minimization was addressed directly by a simulated annealing heuristic algorithm \cite{pettersen2026bipartitioninggraphstatesdistributed}, however, was limited to the case of bisection ($k=2$) and did not discuss how that metric would correspond to a graph state generation protocol for values of $k$ greater than 2.

To provide a scalable protocol-algorithm pair for arbitrary values of $k$ that outperforms the state-of-the-art for edge minimization heuristics across multiple metrics, we start by describing our protocol, VCG, which has a more tractable graph algorithmic objective to visualize than cut rank. Specifically, as also stated in \cite{hein2006}, the minimum vertex cover will upper bound the cut rank, and finding a minimum vertex cover has a polynomial time solution if the graph is bipartite by Kőnig's theorem \cite{konigtheorm, hopcroft_karp}, which states that the minimum vertex cover is equal to the maximum matching in bipartite graphs. Furthermore, for sparse graphs (where sparsity in this context means that the number of edges in every induced subgraph is at most a constant larger than the number of nodes in the subgraph), the matching will roughly equal the rank \cite{konrad_et_al}. 

We can work with bipartite graphs by forming the edges between each pair of QPUs (hamlets) pair-by-pair, recalling that local operations are free, and the intra-QPU edges can be generated after the inter-QPU edges have been established. This assumption also implies that the quantum network being used only ever needs to be able to generate two-qubit inter-QPU GHZ states. 

\begin{figure*}
  \includegraphics[width=\textwidth]{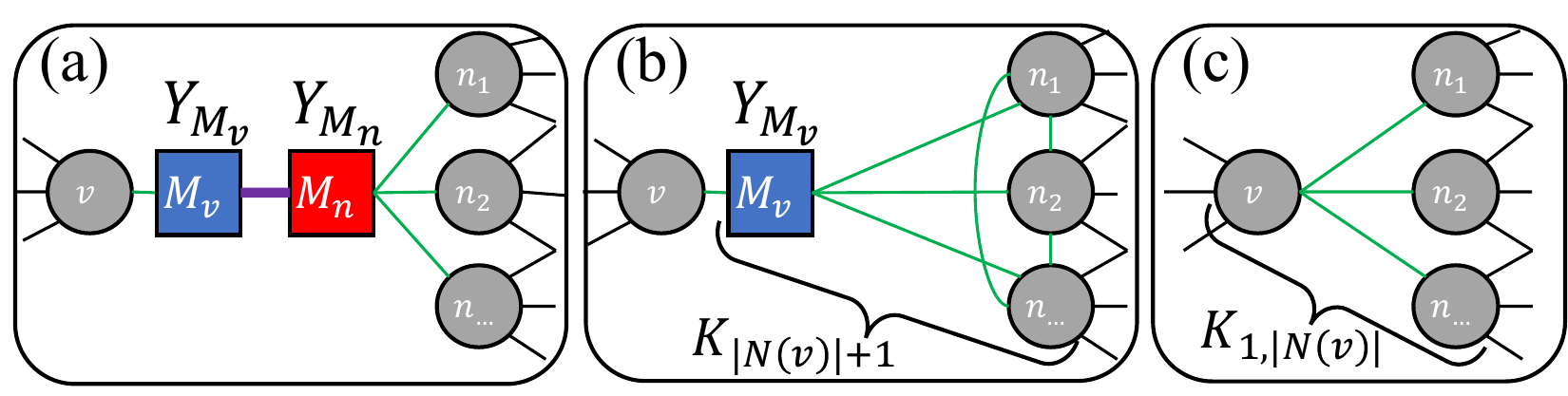}
  \caption{\textbf{Grafting edges to a node of the vertex cover during VCG.} When generating the edges between a vertex of the vertex cover $v$ and its neighbors $n_1,n_2,n_{...}$ on a separate hamlet, $v$ is coupled to its mayor $M_v$, and its neighbors are coupled to theirs, $M_n$. Long-range entanglement is used to entangle the two mayors. This is shown in \textbf{(a)}. \textbf{(b)} shows that Y measurement on $M_n$ creates a clique in the subgraph formed by $M_v$ and $n_1,n_2,n_{...}$. \textbf{(c)} shows that Y measurement on $M_v$ yields our desired result of generating edges between $v$ and its remote neighbors. Notice that this procedure does not affect any preexisting edges that have been generated adjacent to $v,n_1,n_2,n_{...}$}
  \label{fig:vcg}
\end{figure*}

\subsection{Vertex Covering Grafting (VCG)\label{sec:vcg}}
Here, we describe a protocol (VCG) for generating any arbitrary graph state on a collection of homogeneous hamlets (QPUs) by applying $Y$ measurements to the communication qubits (mayors). 
Importantly for intuition in heuristic algorithm design, this protocol uses a number of long range entanglements exactly equal to the size of the matchings between partitioned regions of the graph state. 
For VCG, we are given an already $k$-partitioned graph state, where each partition refers to the set of nodes that will coexist in the same hamlet. We assume one mayor (communication qubit) per hamlet. For ease of description, we consider constructing edges between pairs of hamlets one-by-one; however, parallelism is not prohibited by our approach. For each of the  ${k\choose 2}$ pairs of hamlets, if there are any edges between them, we do the following:

\textbf{1.) Remove intra-hamlet edges} Since we assume local operations to be free relative to the cost of entangling the mayors of two hamlets, edges internal to a hamlet are ignored and expected to be constructed by local operations after all inter-hamlets edges have been generated by VCG.

\textbf{2.) Calculate the minimum vertex cover} As we are only considering two hamlets at a time, ignoring local edges will result in a bipartite graph. This then allows for efficient computation of the minimum vertex cover of that graph by finding the maximum matching \cite{hopcroft_karp}. Once we have the vertex cover, we use one long range Bell pair per node in the cover. As discussed in Section~\ref{sec:background_math}, an $n$ vertex $K_{1,n-1}$ biclique can be generated over a cut for the cost of a single long range Bell pair (see Figure~\ref{fig:vcg}). Treating the nodes of the vertex cover as the roots of these depth 1 trees, we can generate all edges in the graph for a number of Bell pairs equal to the size of the cover. This process is then repeated for all remaining hamlet pairs.

Thus, using this protocol, algorithms that partition graph states to minimize the matching sizes between partitions will directly reduce the number of required long-range Bell pairs. Furthermore, the maximum matching across a cut is by definition a subset of edges of those edges, and therefore is never larger than the total number of cut edges. Moreover, in a bipartite graph $G(V,E)$ with independent vertex sets $A$ and $B$ and $V = A \cup B$, a maximum matching cannot include more edges than $\min(|A|,|B|)$.

\section{Algorithm for Quantum Villager Placement \label{sec:algorithm}}
Here we present a heuristic algorithm for partitioning an arbitrary graph into $k$ pieces, such that the sum of the sizes of the maximum matchings between each pair of partitions is minimized. We show in Appendix ~\ref{sec:npcompleteness}, the problem of determining whether or not a partition of a given graph that yields a provided matching size exists is NP complete, thus motivating our design of a heuristic algorithm. Solving this problem directly reduces the number of Bell pairs required by VCG, as well as indirectly minimizing the cut rank (refer to Appendix~\ref{sec:cutrank}).

\begin{algorithm}[H]
\caption{BURY for Minimizing Matching Cuts}\label{alg:bury_heuristic}
\begin{algorithmic}[H]
    \Require An arbitrary graph $G= (V,E)$ and a resource amount $r_0$, corresponding to the graph state to be partitioned and the number of vertices a partition requires. 
    \Ensure A subset of the vertices of $G$. This subset corresponds to which qubits of the graph state should be on the same processor. \\
    
    \State For vertices $v_1, v_2, ..., v_n$ in $V$, initialize weights $\vec{w} =w_1, w_2,..., w_n$ such that $w_i \gets \deg(v_i)+1$, integer $r=r_0$, and set $C = \emptyset$ to track which nodes have been colored. \\

    \State Let $N(v)$ refer to the neighborhood of $v$.
    \While{$r \neq 0$} 
        \If{$\min(\vec{w})$ < $r$}
            \State $i \gets \text{argmin}(\vec{w})$ 
            \State $v \gets v_i\in V$ 
            \State $r \gets r - w[v]$
            \For{$u \in N(v)\cup v$} 
                \If{$u \notin C$}
                    \State $C \gets C\cup u$ 
                    \State $w[u] \gets w[u] -1$
                    \For{$u' \in N(u)$}
                        \State $w[u'] \gets w[u'] -1$
                    \EndFor
                \EndIf
            \EndFor
            \State $w[v] \gets \infty$
        \Else
            \State Arbitrarily color an uncolored node. This happens when no bury-able node exists. 
            \State $r \gets r -1$
        \EndIf
    \EndWhile
        \State 
    \Return  $C$
\end{algorithmic}
\end{algorithm}

\newpage
Formally, we wish to solve:\\

\textbf{Balanced Minimum Maximum Matching $k$-Partition}: Given a graph $G$, find the balanced $k$-partition of its vertices such that the sum of the sizes of the maximum matchings between partitioned regions (edges within a partitioned region are ignored) is minimized.

More concretely and for ease of discussion, we use the phrase  "coloring a graph" in a way that is nonstandard. That is, we say that given an arbitrary graph $G$ with $n$ vertices and a positive integer $k$, we want to find the coloring such that:

\begin{itemize}
    \item $k$ colors are used as equally as possible throughout the graph. The loosening of this assumption is discussed in Section~\ref{sec:dynamic_compilation}. Each vertex is required to be colored.
    \item Edges incident on nodes of the same color will be deleted, as these correspond to operations local to a QPU, and do not require a quantum network. After all of a hamlet's inter-QPU edges are generated, then these "deleted" edges will be generated by operations local to the hamlet. Temporarily not considering these edges gives us a $k$ partite graph.
    \item The sum of the sizes of the maximum matchings between each pair ($k\choose2$ total) of colors should be minimized after the edge deletion step.  
\end{itemize}

We start by addressing the case where $k=2$. We call our heuristic the "bury heuristic" or BURY, and it is based upon the idea that if a vertex and all of its neighbors are the same color, that vertex will never be able to participate in any matching. Correspondingly, we say that "burying" a vertex $v$ means assigning $v$ and its neighborhood $N(v)$ the same color. To do this, we assign weights to each vertex in the graph equal to the number of nodes required to bury it. For an instance of the heuristic, an amount of resources $r$ (number of nodes that one can color) also needs to be specified. At the beginning of execution of the algorithm, $r$ is usually initialized to $\frac{n}{k}$, as we are naturally aiming to color $\frac{n}{k}$ vertices in the graph. Using our BURY framework, a straightforward heuristic approach is to, at each iteration, bury the vertex with the lowest cost, each time subtracting that cost from the remaining resource value. When no unburied vertices remain with a bury cost less than $r$. Arbitrary nodes can be colored, each time reducing $r$ by 1. In practice, we suspect that sometimes we will find better partitions if these "arbitrary" nodes are adjacent to a node that has already been colored. The greedy heuristic application of the "bury" idea is further detailed in Algorithm~\ref{alg:bury_heuristic}, and an example first iteration on a grid graph is presented in Figure~\ref{fig:step_1}. An implementation in the Julia programming language is available on GitHub, as QuantumHamlets.jl \footnote{\url{https://github.com/QuantumSavory/QuantumHamlets.jl}}~\cite{micciche_2026_18459303}.

\begin{figure}
    \centering
    \includegraphics[width=\linewidth]{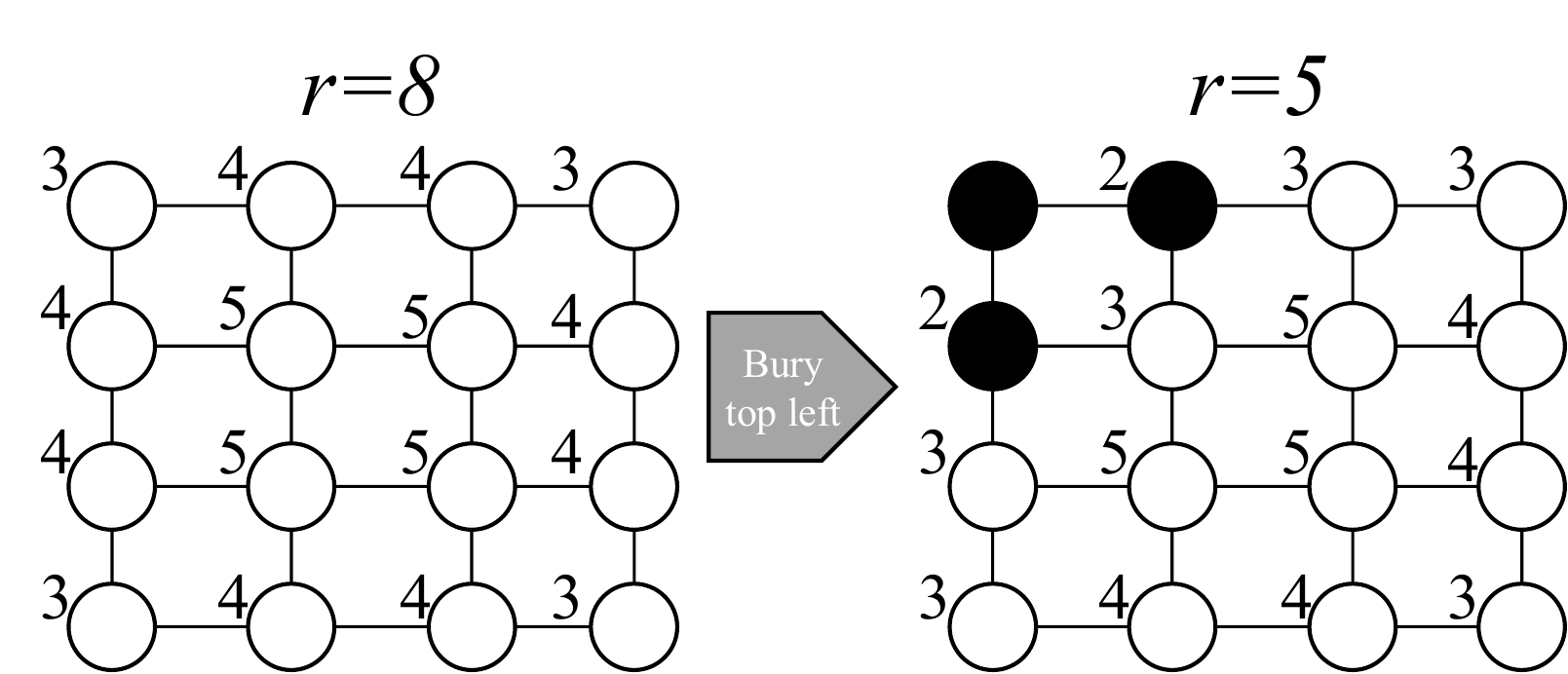}
    \caption{\textbf{Example first step of the bury heuristic algorithm on a $4 \times4$ grid graph.} On the left, the weights of the nodes are initialized to be equal to their degree plus one. The right shows the state of the algorithm after the first iteration on this graph. The top-left vertex has been chosen due to it's minimum value, and has been "buried". That is, it and its neighbors have been decided to be included in the same partition, or colored black as shown here. The costs to bury other nodes has been updated as well. Note that although the neighbors of the chosen node have been colored, they have not been buried yet, and therefore can be chosen in further iterations of the algorithm.}
    \label{fig:step_1}
\end{figure}

To generalize this "bury heuristic" approach to $k$ partition, we simply run Algorithm ~\ref{alg:bury_heuristic} with $r=\frac{n}{k}$, then take the remaining uncolored subgraph, and run the algorithm again. Repeat until only $\frac{n}{k}$ uncolored vertices remain. At this point, the graph has been partitioned into $k$ pieces. 

It is also worth noting that the framing of burying nodes could very well be utilized in more sophisticated algorithms. Moreover, for every iteration the algorithm takes, the largest possible resulting matching decreases. Being able to make statements about the number of steps that can be taken before $r$ reaches zero might be helpful in computing a performance bound in future work. In addition, the flexibility of our heuristic framework is also applicable to the case of dynamic compilation (Section~\ref{sec:dynamic_compilation}), and we hope to see it expanded beyond the greedy manifestation depicted in Algorithm~\ref{alg:bury_heuristic}, to address more complicated settings such as constrained networks in future work. 

\subsection{Time Complexity of BURY\label{sec:time_complexity}}
Referring to Algorithm~\ref{alg:bury_heuristic}, we can establish an upper bound on the time complexity. First, we notice that initialization takes $O(|V|^2)$ time, where $V$ is the number of vertices in $G$. This is because for each vertex, we must calculate its degree. After this, the main loop will run until $r_0$ reaches $0$. In practice, this will often be roughly $r_0$ divided by the average degree in the graph (unless the minimum degree is larger than $r_0$), however, as an upper bound, $r_0$ will always decrement by at least 1 every iteration, and the largest value $r_0$ can take is $\frac{|V|}{k}$. If we are partitioning the graph into $k$ pieces this will run $k$ times, as we simply need to remove colored vertices from the graph and then reset $r_0$ each time. Within this main loop, we find the minimum weight, taking $O(|V|)$ time. Then we look at the neighbors of a chosen vertex, and there are at most $\Delta$, where $\Delta$ is the maximum degree in the graph.  We also update the neighborhoods of those neighbors as well, giving O($\Delta^2)$ time. Therefore, regardless of $k$, the time complexity of BURY is:

$$O\left(|V|^2+ k\frac{|V|}{k} \left(|V| + \Delta^2 \right) \right) = O\left(|V|^2 + |V|\Delta^2\right)$$

Note that as long as $\Delta\leq\sqrt{|V|}$, this is $O(|V|^2)$.  

\section{Results\label{sec:results}}
Here we present results comparing our BURY algorithm to standard algorithms designed for partitioning graphs for edge minimization. As our work is the first attempt, to our knowledge, to address the $k$-partitioning problem where matchings between regions is the objective to be minimized rather than total edges cut, for comparison methods, we can only apply methods that were intended for minimizing the number of edges. Recall that not only is the matching size a much better upper bound on the  entanglement entropy \cite{hein2006}, but also the protocol for graph state generation we presented in Section~\ref{sec:vcg} requires exactly a number of Bell pairs equal to the sum of the matching sizes.

Out of these classical methods for edge-minimization $k$ partition, we find that METIS \cite{metis} far outperforms the other classical methods and on some classes of graphs, such as 3-regular random graphs, outperforms our heuristic, BURY (although BURY outperforms METIS on $m$ regular graphs when $m \geq 4$). Notably, BURY has considerably the best performance of all methods tested, including METIS, on partitioning MBQC compiled graph states of quantum approximate optimization algorithm (QAOA) \cite{farhi2014quantumapproximateoptimizationalgorithm} circuits corresponding to a single iteration of MAX-CUT on a complete graph. The compiled graph states were generated using Graphix \cite{graphix}.

\subsection{Methods \label{sec:methods}}
Here we enumerate the methods we apply to partitioning a provided graph into $k$ pieces. We evaluate them on how well they reduce the sizes of the maximum matchings between partitions which directly corresponds to the number of inter-QPU Bell pairs required by our protocol VCG. Furthermore, we find that this metric also corresponds to minimizing the cut rank, and is discussed further in Appendix~\ref{sec:cutrank}.

\begin{figure}
    \centering
    \includegraphics[width=\linewidth]{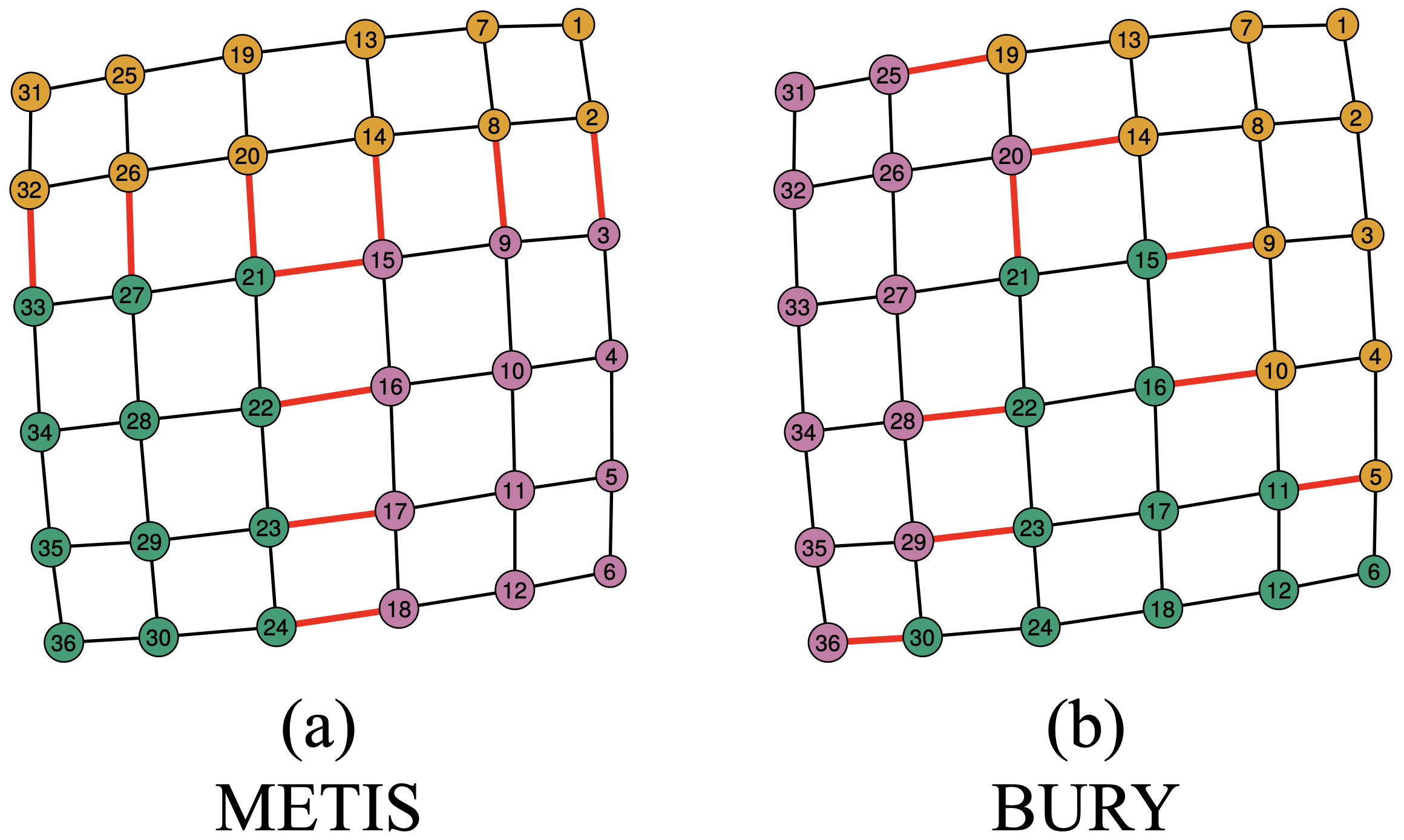}
    \caption{\textbf{Example of performance of METIS and BURY on a $6\times6$ square grid graph, partitioned into three hamlets. (a)} METIS partitions in a way that causes a maximum matching of size 4 between green and pink, 3 between orange and green, and 3 between orange and pink. For a total of 10 Bell pairs when generated by VCG. \textbf{(b)} Our BURY heuristic partitions such that 4 Bell pairs are required between green and pink, 2 between pink and orange, and 3 between orange and green. For a total of 9.}
    \label{fig:grid_ex}
\end{figure}

\begin{itemize}
    \item \textbf{Saran-Vazirani}
    One of the earlier approximation algorithms ($\frac{n}{2}$--approximation) for balanced graph bi-partition \cite{saran_vazirani}. This algorithm can be applied to the balanced $k$ partition problem, by exponential scaling with $k$. For this reason, as well as due to relative poor performance, it is not included in plots of larger graphs. 
    \item \textbf{Random Sampling \textit{N}}
    As a comparison method and baseline, we find the best partition on the metric of matching sizes out of $N$ random partitions.
    \item \textbf{Kernighan–Lin (KL)} \cite{kl_algo} A heuristic algorithm often used for graph partitioning, especially for graph bisection. 
    \item \textbf{METIS} \cite{metis} A state-of-the-art graph partitioning software based on graph coarsening \cite{KARYPIS199896}, designed for unstructured graph partitioning. 
    \item \textbf{BURY}
    Implementation of our heuristic algorithm discussed in Section~\ref{sec:algorithm}, and presented in Algorithm~\ref{alg:bury_heuristic}.
    \item \textbf{BURY - SEEDING}
    A modification of BURY. In addition to the input of BURY, a specific vertex is specified as the "seed" and this seed node starts colored. When selecting to color nodes of minimum weight, this heuristic only looks at nodes adjacent to nodes that have already been colored, whereas BURY chooses a node of minimum weight over the entire graph, regardless if it is adjacent to an already colored node or not.
\end{itemize}

    \begin{figure}
        \centering
        \includegraphics[width=\linewidth]{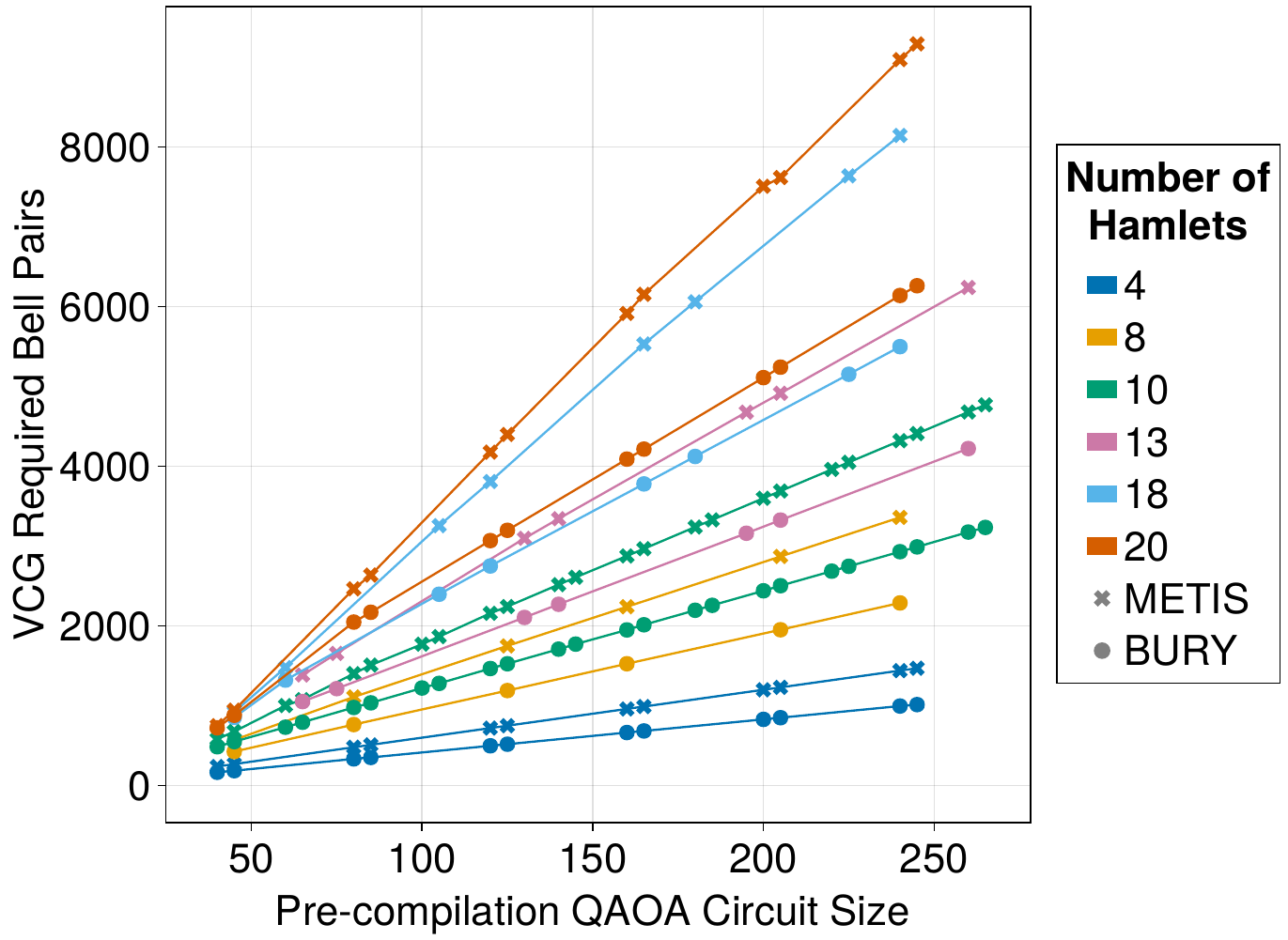}
        \caption{\textbf{Performance of METIS and BURY on compiled QAOA graph states} The $y$-axis corresponds to the number of Bell pairs required to generate the distributed graph state using our VCG protocol. This value is also exactly equal to the sum of the maximum matching sizes between all partitions. The $x$-axis corresponds to the size of the QAOA circuit the MBQC graph state was compiled from. The colors correspond to varying the number of hamlets while the keeping the number of total logical qubits across all hamlets constant for a given $x$-value. The marker type distinguishes our BURY heuristic from the METIS algorithm.}
        \label{fig:qaoa}
    \end{figure}
\subsection{Evaluated Graphs}
Here we enumerate and discuss several of the graphs on which we tested the above methods:
\begin{itemize}
    \item \textbf{Compiled QAOA graph states}
    Using the Python package Graphix \cite{graphix}, we generate the compiled graph state corresponding to single iteration QAOA circuits for MAX-CUT on a complete graph. This is presented in Figure~\ref{fig:qaoa}. Note that the $x$ axis is the number of qubits in the original QAOA circuit. After converting these circuits to the MBQC model and then optimizing those graph states with Graphix, the number of vertices in the graph is considerably greater than the value on the $x$ axis. For example, a QAOA circuit with 5 qubits compiles to a 20 vertex graph state, while 230 qubits compiles to a 26,795 vertex graph state.
    
    As shown in Figure~\ref{fig:qaoa}, our heuristic algorithm performs better at partitioning graph states for generation by our VCG protocol than METIS, and this gain in performance is exacerbated at larger graph state sizes and at larger numbers of partitions (more hamlets/QPUs). This performance difference also holds when evaluated on cut rank, as all graphs tested do (Figure~\ref{fig:qaoa_cutrank}).

    \item \textbf{Grid graphs} Grid graphs refer to the square lattice graphs. For vertex sizes that are not perfect squares we choose the length and width dimensions to be as close to equal as possible. Refer to Figure~\ref{fig:grid_ex} for an example showing the difference between METIS and BURY partitions. In Figure~\ref{fig:grid}, we compare the performance of METIS and BURY on grid graphs of sizes ranging from 25 to 1,600, and partitioning it into from 2 to 10 partitions.

    \begin{figure}
        \centering
        \includegraphics[width=\linewidth]{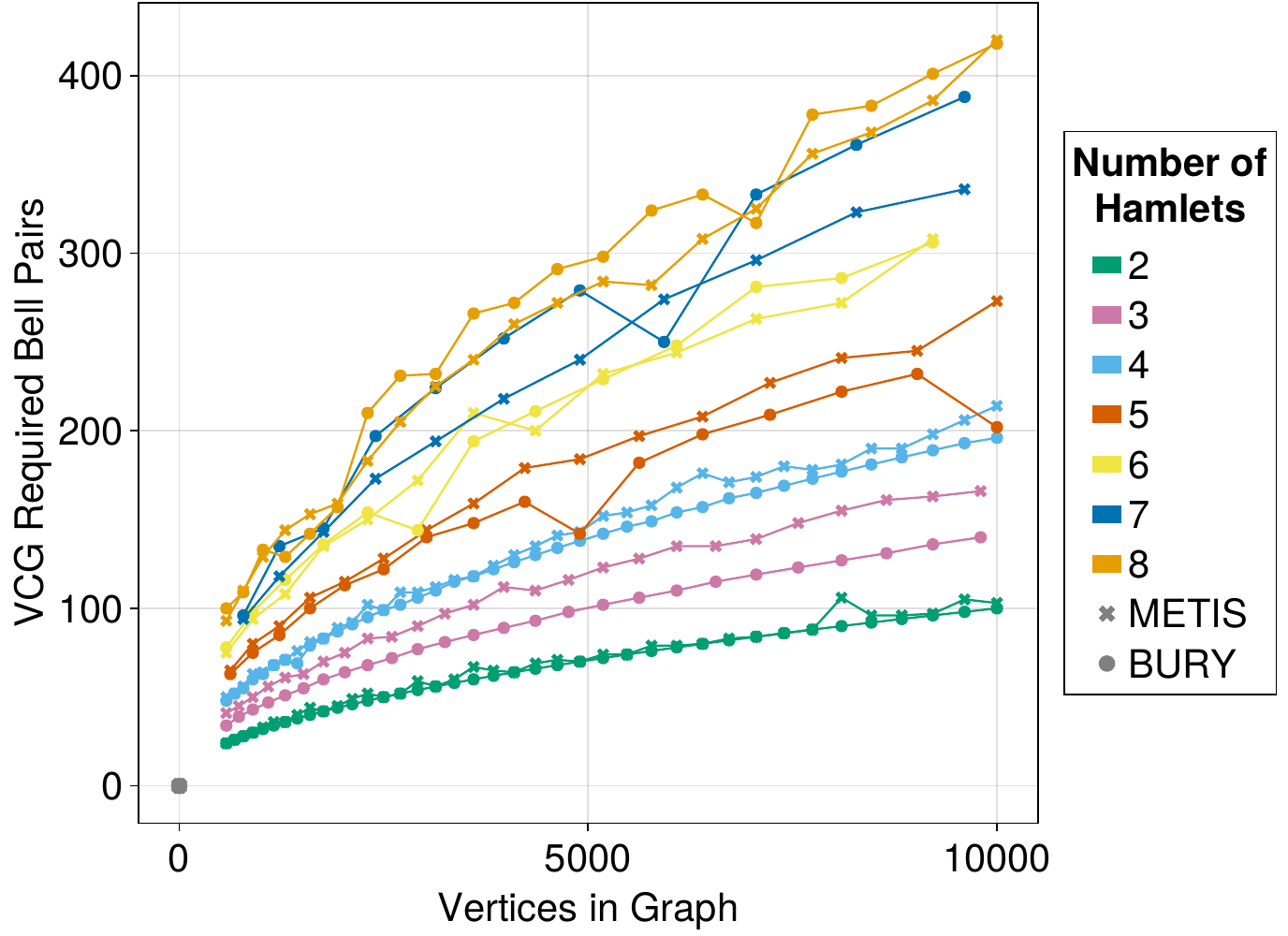}
        \caption{\textbf{Performance of METIS and BURY on square grid graphs} The $y$-axis corresponds to the number of Bell pairs required to generate the distributed graph state using our VCG protocol. The $x$-axis corresponds to the total number vertices in the graph being partitioned. The different colors correspond to increasing the number of hamlets and thus doing larger values of $k$ in maximum matching $k$-partition. For a given $x$ value, the number of vertices in the total graph is held constant when increasing the number of hamlets. In other words, more hamlets corresponds to fewer logical qubits per QPU.}
        \label{fig:grid}
    \end{figure}

    For most grid graph sizes and number of hamlets tested, BURY outperforms METIS in minimizing the number of required Bell pairs. However, METIS does occasionally outperform BURY on certain sizes of grid, for certain values of $k$.
    
    \item \textbf{Random regular graphs}
    Using the function defined in Graph.jl \cite{GraphsJL}, we test the methods described in Section~\ref{sec:methods} on various random regular graphs. Figure~\ref{fig:2ham_3reg} shows our BURY algorithm outperforming all but METIS on bipartitioning random 3-regular graphs. Figure~\ref{fig:2ham_6reg} shows BURY outperforming all methods including METIS on bipartitioning random 6-regular graphs, and the same is true for 4- and 5-regular graphs as well, however those plots are omitted for conciseness.  
    
\end{itemize}

\begin{figure}
    \centering
    \includegraphics[width=\linewidth]{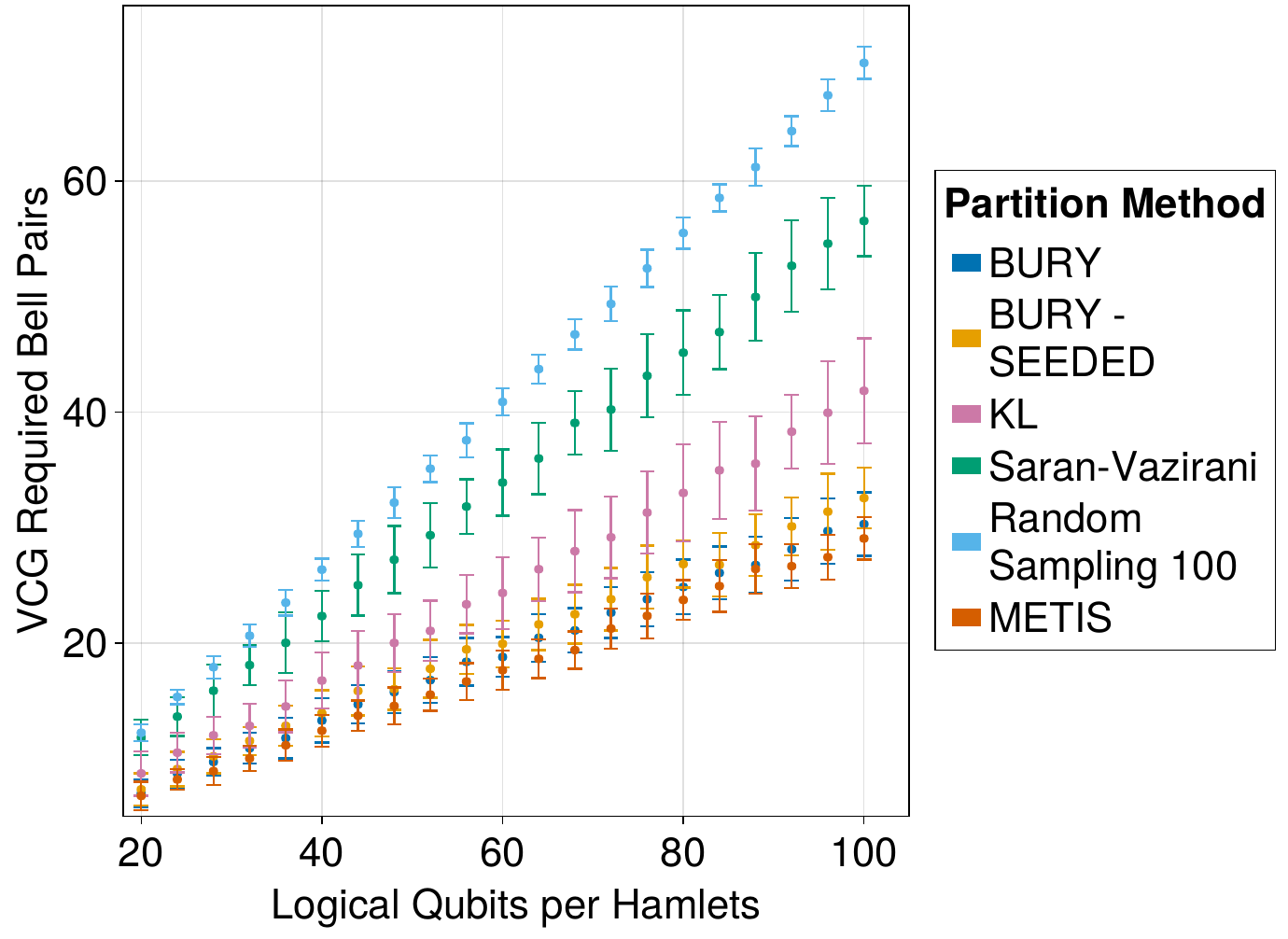}
    \caption{\textbf{Different methods performance on  bipartitioning random 3-regular graphs} The $y$-axis corresponds to the number of Bell pairs required to generate the distributed graph state using our VCG protocol. The $x$-axis corresponds to the number of logical qubits in a hamlet, so the number of total vertices is twice the $x$-axis value. Because we generate random graphs, we take the average performance of 50 random samples per data point. }
    \label{fig:2ham_3reg}
\end{figure}

\begin{figure}
    \centering
    \includegraphics[width=\linewidth]{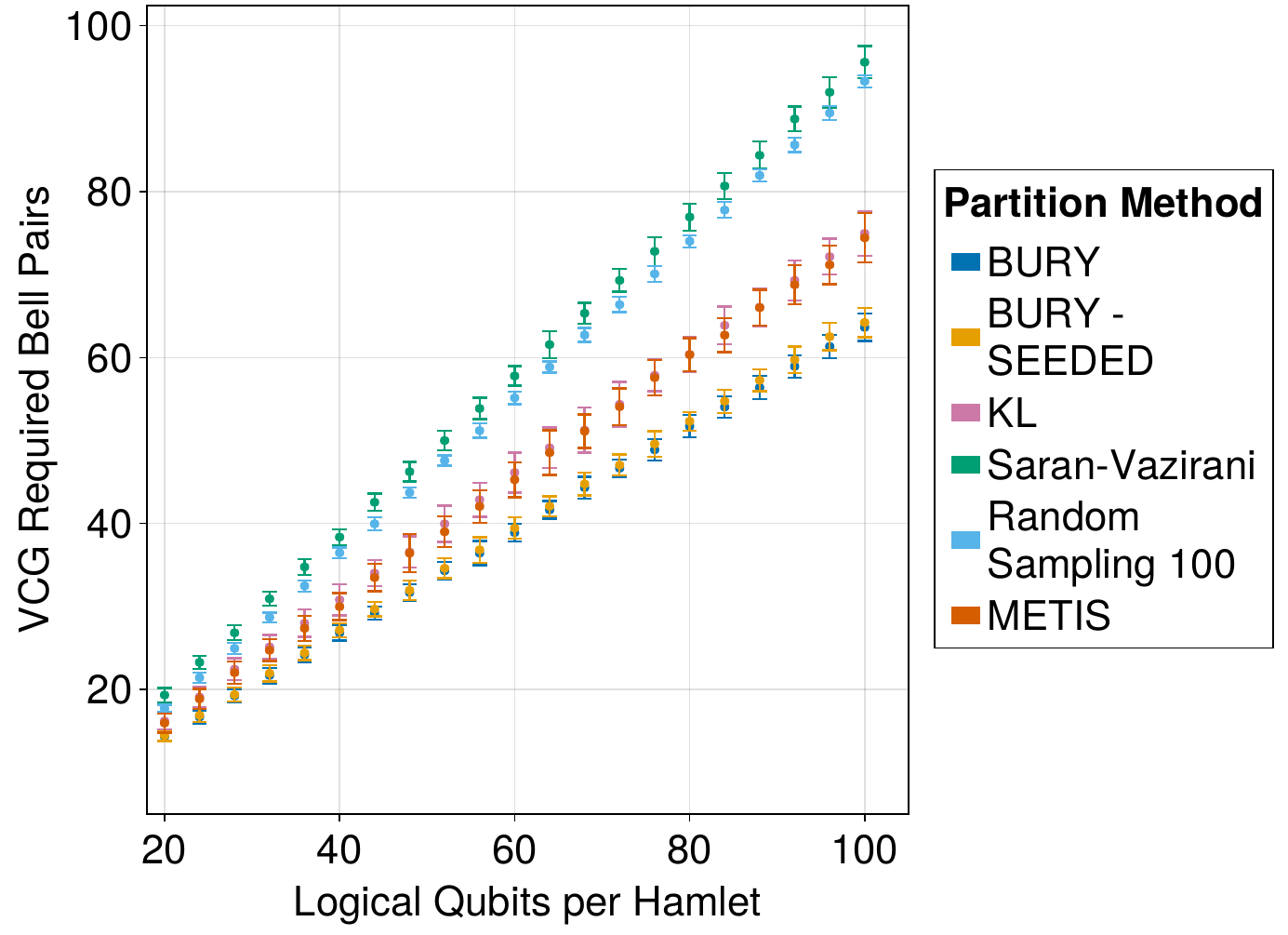}
    \caption{\textbf{Different methods performance on  bipartitioning random 6-regular graphs} Same as Figure~\ref{fig:2ham_3reg}, but for random 6-regular graphs instead of 3-regular.}
    \label{fig:2ham_6reg}
\end{figure}

\section{MBQC and Dynamic Compilation\label{sec:dynamic_compilation}}
So far, we have considered the problem of statically partitioning a large graph state over $k$ many hamlets. Here, statically refers to generating the entire graph state at once, utilizing all resources, and then starting to measure it after it has been generated in its entirety. Although this problem alone is relevant whenever a distributed graph state present at a single time is desired, this is strictly less practical in the MBQC model. Here we discuss how our techniques are also relevant and applicable to the dynamic compilation problem, that is, the problem where the graph state is measured during its generation, as was often mentioned  historically when considering cluster states for computation~\cite{one_way}. 

One element of MBQC that has been notably missing so far from our problem statement is consideration of the  measurement order of qubits in the graph state we want to generate. Although we leave true integration of concepts such as flow \cite{determinism_in_one_way} and generalized flow \cite{g_flow} for future work, we provide a few statements on how our work can naturally be adapted to the setting of measurement orders. For each vertex $v$ in the graph state, we consider $v$ to be in one of the $\Lambda$ disjoint sets $L_1,L_2,...,L_{\Lambda}$, and $\bigcup_{i=1}^\Lambda L_i = V$. The $l$th layer $L_l$ is a set of vertices that may be measured in any order with respect to one another. How these layers relate to each other under different assumptions and interpretations is discussed in the following subsections.

\subsection{Single layer case}
By definition, the case where we only have a single layer is exactly the static compilation problem. 

\subsection{Rigid layer case}
When interpreting layers in a stricter sense ("rigid layers"), then all $v\in L_l$ must be measured before any $u \in L_{l+1}$, and the graph-state defined by the subgraph induced by $G[L_l\cup L_{l+1}]$ must have been generated before measurement. In this interpretation, $G[L_1\cup L_{2}]$ needs to be generated first, then $L_1$ is measured, then $L_3$ is generated, then $L_2$ is measured, and so on. Generation of $G[L_1\cup L_{2}]$ is the homogeneous partition problem, assuming that $|L_1\cup L_{2}|>H$, where $H$ is the number of of villagers in one of the hamlets. If this does not hold, then $G[L_1\cup L_{2}]$ can fit in a single hamlet, and partitioning of the graph state is not needed.

After measuring $L_1$, $G[L_2\cup L_{3}]$ must be generated next. However, $L_2$ will have already been generated. Therefore, the problem assumption of homogeneous hamlets is no longer valid. Rather, subtracting the number of already generated nodes present in each hamlet from that hamlet's capacity gives the heterogeneous problem.

In the heterogeneous problem, hamlets may be differing sizes. While this work does not provide further analysis of this specific problem, we note that our BURY algorithm (Algorithm~\ref{alg:bury_heuristic}) will easily be able to adapt to this by setting $r$ equal to the remaining capacity of a given hamlet and pre-coloring the appropriate nodes in the weight table to reflect, in essence, a homogeneous problem instance of BURY that is starting mid-execution. We naively assume that a greedy approach of always trying to fill the hamlet with most vacancies will work well. 
 
\subsection{Non-rigid layer case\label{sec:non_rigid}}
The rigid way of viewing layers is naturally not optimal in terms of the spatial cost of the number of qubits required to do computation, as a vertex $v\in L_l$ may be measured as soon as all of its neighbors in $L_{l-1}$ have been measured, and all of its neighbors in $L_{l+1}$ have been generated. However, the rigid case always obeys the constraints of the non-rigid case, and the rigid case provides a framework where our static partition algorithms can be applied directly. Further work on this distributed dynamic compilation in this non-rigid measurement schedule regime we believe to be of great interest for future work. Recently, approximating the spatial cost of a graph was proven to be NP-hard \cite{Elman_2025}, implying that there is much work to be done on heuristic methods and for specific cases of graphs. 
 
\section{Conclusion\label{sec:conclusion}}
We have explored the problem of graph partitioning with matching minimization in mind rather than edge minimization, and designed a heuristic algorithm, BURY, to address it. Our method far outperforms existing edge minimization algorithms for the purposes of minimizing the amount of entanglement between graph partitions of a graph state both in terms of matching sizes and cut rank. Furthermore, we provided an explicit protocol for generating a given partition graph state for a number of inter-QPU Bell pairs equal to the matching sizes. Moreover, we hope that this work leads to many more applications of the BURY heuristic framework to more nuanced situations. Such as adapting BURY to constrained networks, as this work assumes all-to-all connectivity in the quantum network.

Other potential future work that would be of great use would be, as mentioned in Section~\ref{sec:non_rigid}, further design and analysis of schedulers, as well as compilers, for the dynamic case of graph state partitioning with arbitrary measurement schedules. In addition a study of the graphical properties of the compiled graph states for various well-known quantum algorithms and quantum algorithmic primitives would be helpful for design of better heuristic methods and schedulers. More efficient protocols than VCG for generating graph states and their corresponding compilation algorithms would also be very fruitful. Since METIS performed so well, modification of METIS and other graph coarsening based graph partitioning algorithms might prove fruitful for graph state partitioning.

\paragraph*{Code availability:} The underlying simulator is available as an open source package in the Julia ecosystem: QuantumHamlets.jl~\cite{micciche_2026_18459303}. The goal of this package is to allow users to test various algorithms for partitioning graphs across the multiple metrics discussed in this paper. Naturally, an implementation of BURY is also provided. 

\begin{acknowledgements}
We would like to thank Simon Devitt for discussion related to graph state compilation. We also thank Kenneth Goodenough and Guus Avis for fruitful discussions concerning previous work in graph state generation in quantum networks. This work was supported by the JST Moonshot R\&D program under Grant JPMJMS226C and by NSF grants 1941583, 2346089, 2402861, 2522101, 2521579. 
\end{acknowledgements}

\bibliography{bibliography}

\appendix

\section*{Appendix}
\subsection{Performance on Cut Rank\label{sec:cutrank}}
Here we show provide plots and discuss how our BURY algorithm also outperforms existing methods for $k$ partition when evaluated on the metric of cut rank.

First, to define cut rank, we consider a cut as a subset of edges $C\subseteq E$ of some graph $G = (V,E)$. Then, we consider the adjacency matrix, $A$, representing edge-induced subgraph $H = (V', C)$, where $V' \subseteq V$ is the set of vertices incident to the edges in $C$. $H$ is naturally bipartite from this construction, so the adjacency matrix will have the following form:
$$A = 
\begin{pmatrix}
0 & B\\
B^T & 0
\end{pmatrix}$$

Where $B$ is the biadjacency matrix. We call the rank of $B$ over the field of binary numbers the \textit{cut rank} of the partition. As mentioned in Hein et al. \cite{hein2006}, the rank of this matrix is a tighter upper bound on the entanglement entropy across a cut, and is likely the best one can hope for in terms of the number of Bell pairs required to generate a graph state, regardless of how clever the protocol. Furthermore, for the case of free local CZ gates, and in the case of $k=2$, a generation protocol exists based on Gaussian elimination to generate a graph state with a number of Bell pairs equal to the cut rank by Theorem 1 of Fattal et al.\cite{fattal2004entanglementstabilizerformalism}. Our BURY algorithm minimizing cut rank as well as maximum matchings is evidence that BURY is a good algorithm for graph state partitioning regardless of a  graph state generation method used.

We show in Figures~\ref{fig:qaoa_cutrank},~\ref{fig:grid_cutrank},~\ref{fig:2ham_3reg_cutrank}, and~\ref{fig:2ham_6reg_cutrank} that BURY performs just as well at minimizing cut rank as it does for minimizing matching sizes between partitions, as described in Section~\ref{sec:results}. 

    \begin{figure*}
        \centering
        \includegraphics[scale=0.55]{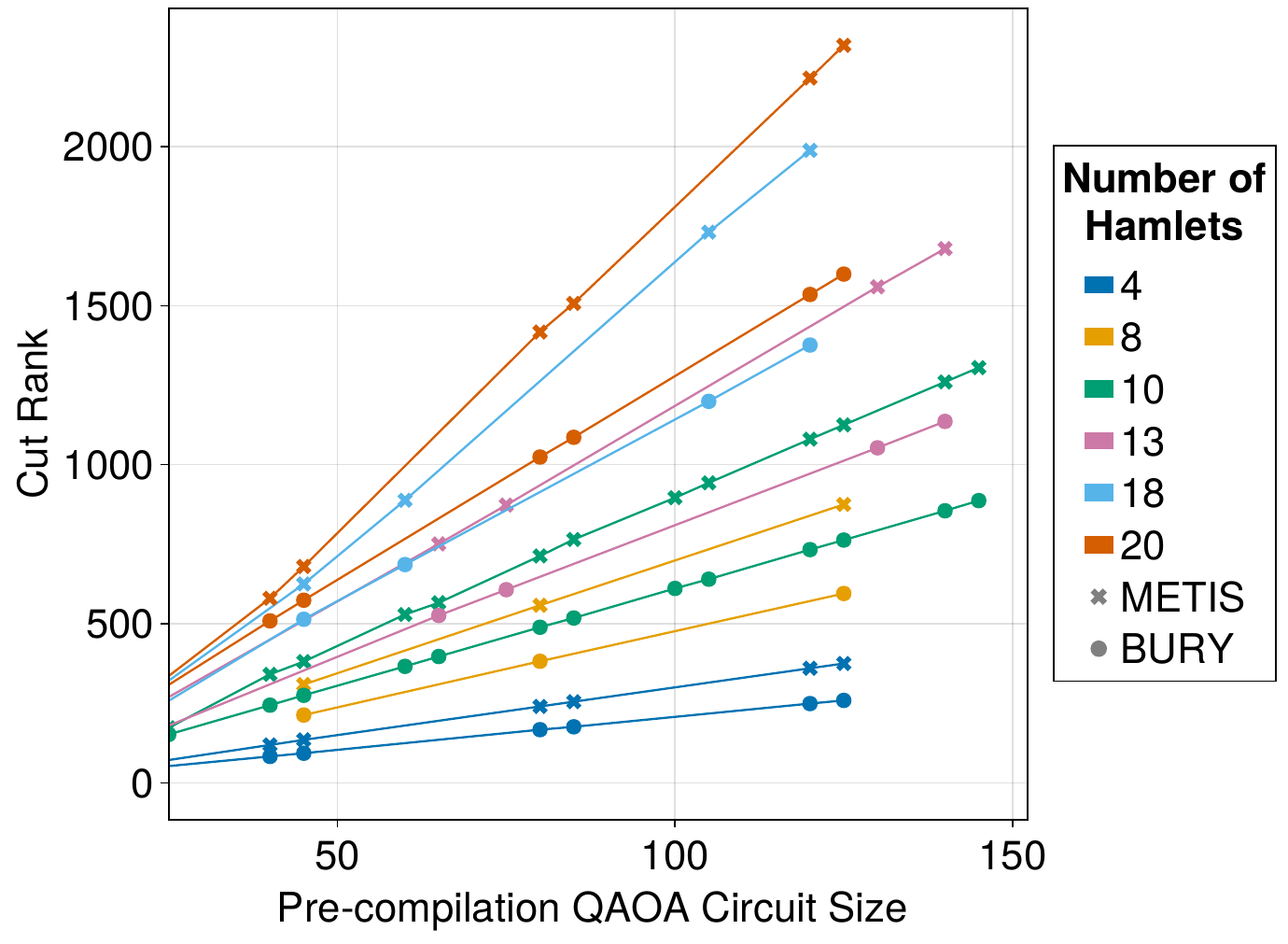}
        \caption{\textbf{Performance of METIS and BURY on compiled QAOA graph states} The $y$-axis corresponds to the sum of the ranks of the biadjacency matrices corresponding to each cut in the $k$ partition. The $x$-axis corresponds to the size of the QAOA circuit the MBQC graph state was compiled from. The colors correspond to varying the number of hamlets while the keeping the number of total logical qubits across all hamlets constant for a given $x$-value. The marker type distinguishes our BURY heuristic from the METIS algorithm.}
        \label{fig:qaoa_cutrank}
    \end{figure*}

    \begin{figure*}
        \centering
        \includegraphics[scale=0.55]{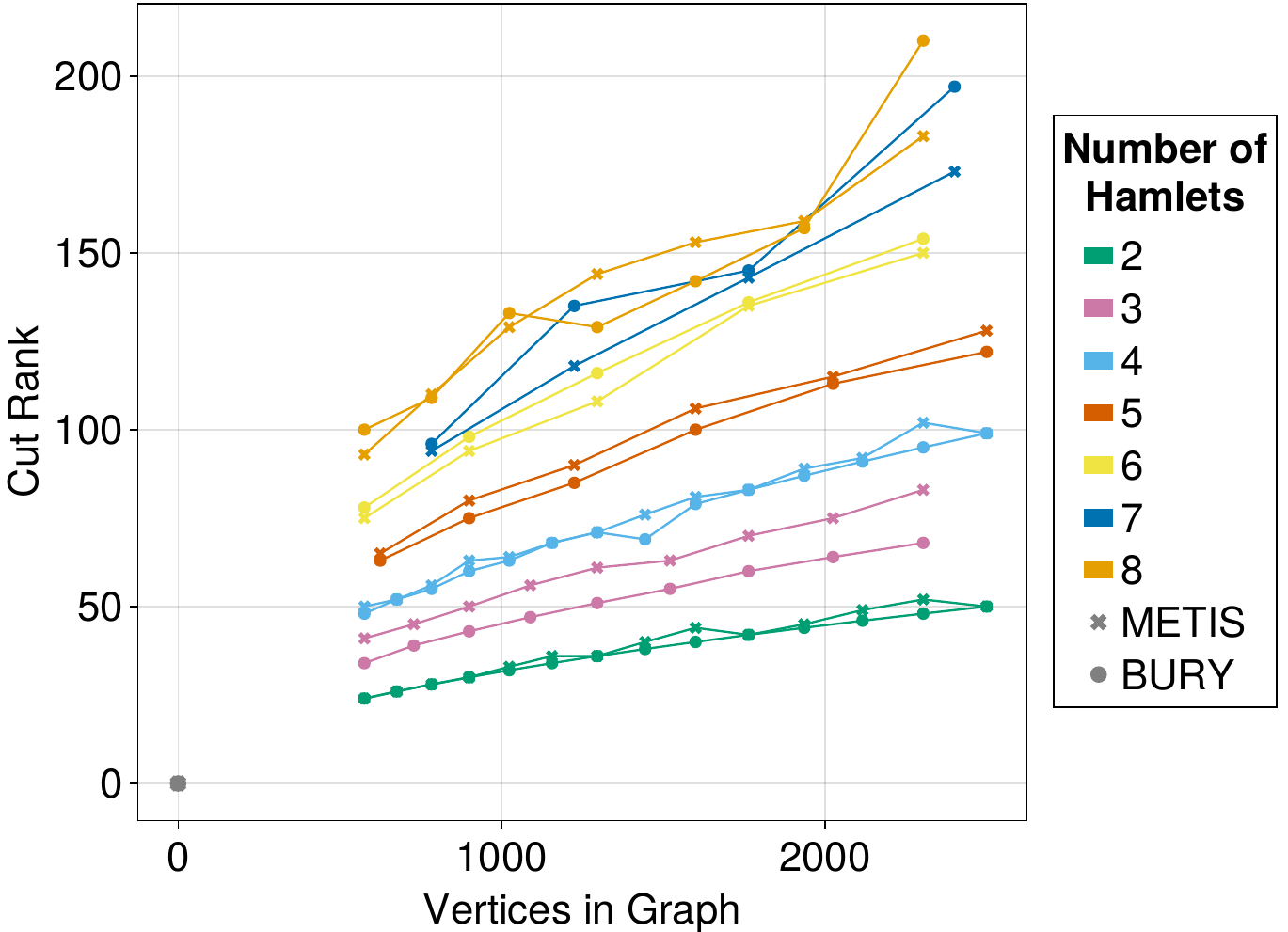}
        \caption{\textbf{Performance of METIS and BURY on square grid graphs} The $y$-axis corresponds to the sum of the ranks of the biadjacency matrices corresponding to each cut in the $k$ partition. The $x$-axis corresponds to the total number vertices in the graph being partitioned. The different colors correspond to increasing the number of hamlets and thus doing larger values of $k$ in maximum matching $k$-partition. For a given $x$ value, the number of vertices in the total graph is held constant when increasing the number of hamlets. In other words, more hamlets corresponds to fewer logical qubits per QPU.}
        \label{fig:grid_cutrank}
    \end{figure*}

    \begin{figure*}
    \centering
    \includegraphics[scale=0.55]{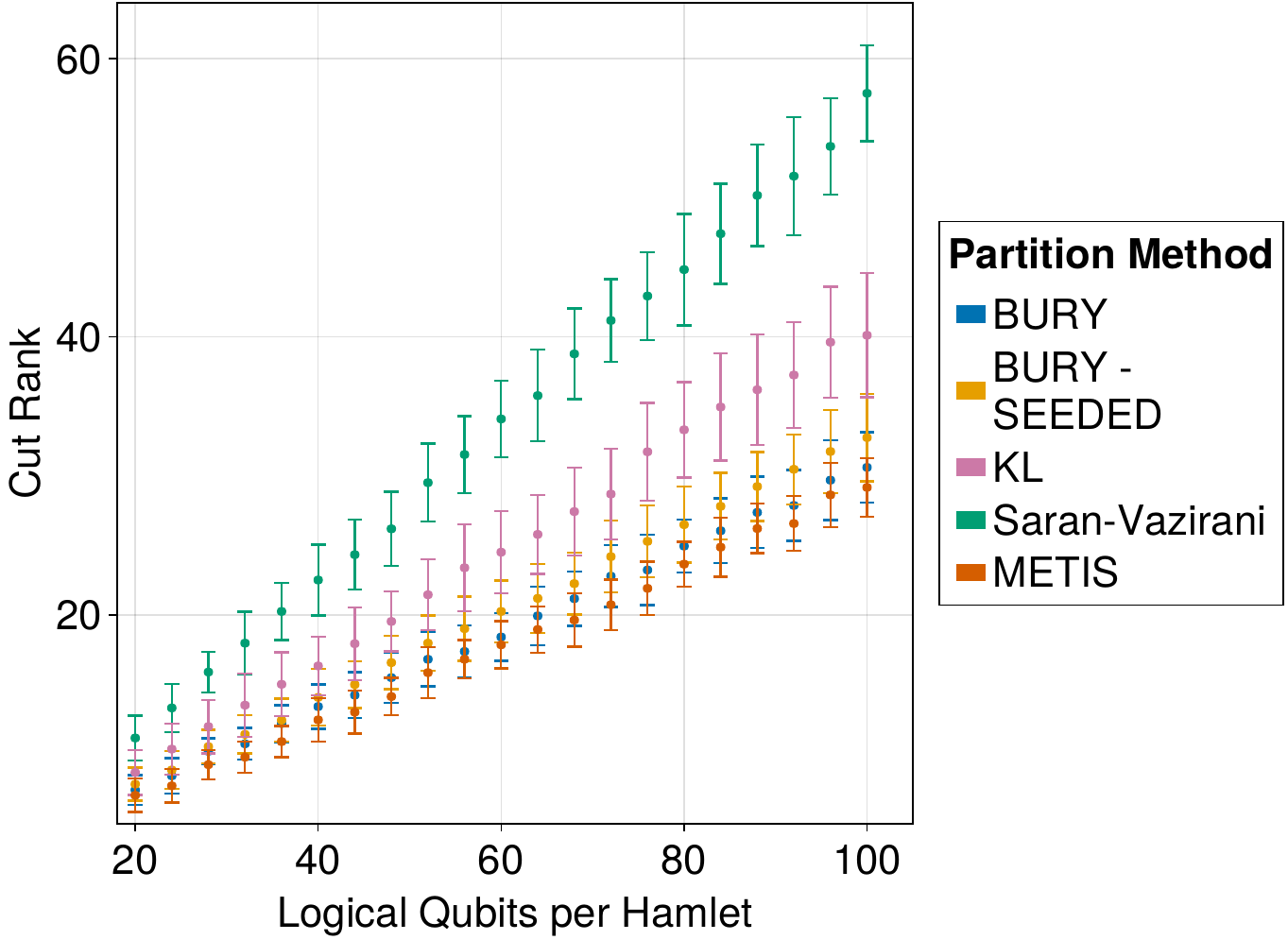}
    \caption{\textbf{Different methods performance on  bipartitioning random 3-regular graphs} The $y$-axis corresponds to the sum of the ranks of the biadjacency matrices corresponding to each cut in the $k$ partition. The $x$-axis corresponds to the number of logical qubits in a hamlet, so the number of total vertices is twice the $x$-axis value. Because we generate random graphs, we take the average performance of 50 random samples per data point. }
    \label{fig:2ham_3reg_cutrank}
\end{figure*}

\begin{figure*}
    \centering
    \includegraphics[scale=0.55]{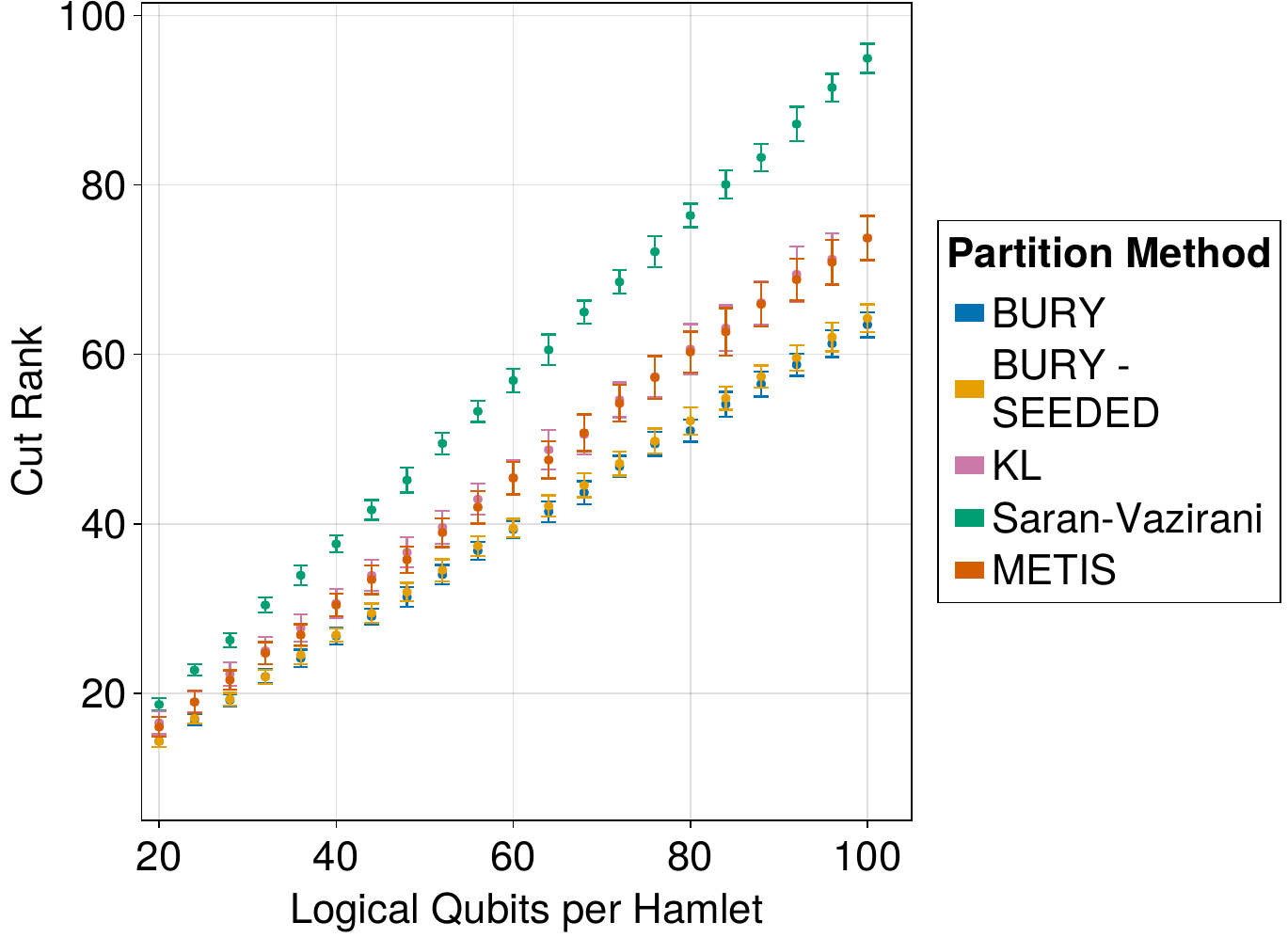}
    \caption{\textbf{Different methods performance on  bipartitioning random 3-regular graphs} The $y$-axis corresponds to the sum of the ranks of the biadjacency matrices corresponding to each cut in the $k$ partition. The $x$-axis corresponds to the number of logical qubits in a hamlet, so the number of total vertices is twice the $x$-axis value. Because we generate random graphs, we take the average performance of 50 random samples per data point. }
    \label{fig:2ham_6reg_cutrank}
\end{figure*}

\subsection{NP Completeness\label{sec:npcompleteness}}

\newcommand{\MM}{\nu}

We consider the following three decision problems:
\begin{enumerate}
\item {\em Minimum Balanced Cut-Matching (BCM).}
Given an undirected graph $G=(V,E)$ integer $k$, does there exist a partition $V=A\cup B$ such that $|A|=|B| =|V|/2$ and $\MM(E(A,B)) \leq k$
where $E(A,B)$ are the edges crossing $(A,B)$ and $\nu(\cdot)$ denotes maximum matching size.
\item {\em Minimum Balanced Vertex Separator (BVS).} 
Given an undirected graph $G=(V,E)$ integer $k$, 
does there exist a partition $V=X\cup Y\cup U$ such that $|U|\leq k$, $\max(|X|,|Y|)\leq |V|/2$ and there are no edges between $X$ and $Y$.
\item {\em Minimum Balanced Cut (BC).} Given an undirected graph $G=(V,E)$ and integer $k$, does there exists a partition $V=A\cup B$ with $|A|=|B|=|V|/2$ such that $|E(A,B)|\leq k$.
\end{enumerate}

\begin{lemma}
$(G,k)$ is a ``yes'' instance of BCM iff $(G,k)$ is a ``yes'' instance of BVS.
\end{lemma}

\begin{proof}
Suppose $(G,k)$ is a ``yes'' instance of BVS. Then there exists partition $V=X\cup Y\cup U$ with no edges between $X$ and $Y$, $|X|,|Y|\le |V|/2$ and $|U|\leq k$. Partition $U$ into $U_1\cup U_2$ such that $A:=X\cup U_1$ and $B:=Y\cup U_2$ both have size $|V|/2$. Next, observe that $\MM(E(A,B))\leq |U_1|+|U_2|= |U|\leq k$ because there are no edges between $X$ and $Y$. Hence $(G,k)$ is a ``yes'' instance of BCM as required.

Now suppose $(G,k)$ is a ``yes'' instance of BCM. Then there exists a partition $V=A\cup B$ with $|A|=|B|=n/2$ such that $\MM(E(A,B))\leq k$. By K\H{o}nig's theorem, the minimum vertex cover of a bipartite graph has the same size as the maximum matching. Hence, there exists $C\subset V$ with $|C|\leq k$ such that every edge in $E(A,B)$ has at least one endpoint in $C$. Setting $X=A\setminus C$ and $Y=B\setminus C$ ensures the partition $V=X \cup Y \cup C$ has the necessary properties to ensure $(G,k)$ is a ``yes'' instance of BVS.
\end{proof}

Therefore the NP-hardness of BCM would follow from the NP-hardness of BVS. The NP-hardness of  BVS is implied by the NP-hardness of BC \cite{GAREY1976237}. While this last step of the argument is fairly standard (to the extent that some authors just cite work that establishes BC is NP-hard without being explicit about the reduction), we include a reduction from BC to  BVS for completeness. Specifically, given an instance $(G,k)$ of BC we construct an instance $(H,k)$ of BVS where $H$ is a graph constructed as follows:
\begin{itemize}
\item For each $v_i\in G$, generate a set of nodes $V_i=\{u_i^1, \ldots, u_i^t\}$ where $t=3m+k+1$.
\item The nodes of $H$ are $V_1\cup \ldots \cup V_n \cup \{w_{i,j}:\{v_i,v_j\} \in G\}$ and $2m$ isolated nodes. For each $i\in [n]$, include all edges with each $V_i$, i.e., the induced subgraph $H[V_i]$ is a clique for all $i$. Also, for each $\{v_i,v_j\} \in G$, add edges between $w_{i,j}$ and every node in $V_i\cup V_j.$
\end{itemize}

Then if $(G,k)$ is a ``yes'' instance of BC, then consider the balanced partition $(A,B)$ where $|E(A,B)|\leq k$. Let 
\begin{eqnarray*}
U &=& \{w_{i,j}:\{v_i,v_j\}\in E(A,B)\} \\ 
X &=& \left (\cup_{v_i\in A}V_i \right )\cup \{w_{i,j}:\{v_i,v_j\}\in E, v_i,v_j\in A\}\\ 
Y &=& \left (\cup_{v_i\in B}V_i \right )\cup \{w_{i,j}:\{v_i,v_j\}\in E, v_i,v_j\in B\}
\end{eqnarray*}
Then $|U|\leq k$, $|X|=tn/2+m_A$, and $|Y|=tn/2+m_B$ where $m_A$ and $m_B$ are the number of edges in $G[A]$ and $G[B]$ respectively. Distributing $\lfloor m-m_A/2+m_B/2\rfloor $ isolated nodes to $X$ and $ \lceil m -m_B/2+m_A/2 \rceil$ isolated nodes to $Y$, ensures 
\begin{eqnarray*}
\max(|X|,|Y|) & \leq & tn/2+m+m_A/2+m_B/2  \\ 
& \leq & (tn+3m)/2=|V(H)|/2
\end{eqnarray*}
 while ensuring there are no edges between $X$ and $Y$ in $H$. Hence, $(H,k)$ is a ``yes'' instance. 

Conversely if $(H,k)$ is a ``yes'' instance of BVS, let $V(H)=X \cup Y \cup U$ be a partition of $V(H)$ with $|X|,|Y|\leq |V(H)|/2= tn/2+3m/2$, $|U|\leq k$, and no edges between $X$ and $Y$. Let $V'_i=V_i\setminus U$. Note that since $t>k$, no $V'_i$ is empty and must be entirely contained in $X$ or $Y$ since $G[V'_i]$ is a clique and there are no edges between $X$ and $Y$. Let \[A=\{i\in [n]:V_i' \subseteq X\} ~~~\mbox{ and }~~~ B=\{i\in [n]:V_i' \subseteq Y\} \ .\] Note $t |A|-k \leq |\cup_{i\in A} V_i'|$ and $t |B|-k \leq |\cup_{i\in B} V_i'|$, and hence \[\max(t |A|-k , t|B|-k) \leq tn/2+3m/2\] and so \[\max(|A|,|B|) \leq n/2+(3m/2+k)/t \ .\] Since $t>(3m+k)$ and $|A|+|B|=n$, we deduce $|A|=|B|=n/2$.

 Consider any edge $\{v_i,v_j\}\in E(G)$ with $i\in A$ and $j\in B$. The corresponding vertex $w_{i,j}$ is adjacent to every vertex of $V_i\cup V_j$, and in particular to every vertex of the nonempty sets $V_i'\subseteq X$ and $V_j'\subseteq Y$. Therefore $w_{i,j}$ cannot lie in $X$ (or else it would have an edge to $Y$), and it cannot lie in $Y$ (or else it would have an edge to $X$). Thus $w_{i,j}\in U$.
So every edge of $G$ crossing $(A,B)$ corresponds to a distinct vertex of $U$, implying
\[
|E_G(A,B)| \le |U| \le k.
\]
Therefore $(A,B)$ is a balanced cut of $G$ of size at most $k$, and so $(G,k)$ is a ``yes'' instance of BC.

\end{document}